\documentclass{article}
\usepackage{amsfonts}
\usepackage{amssymb}
\usepackage[dvips]{color}
\usepackage{graphicx}
\usepackage{amsmath}

\newtheorem{theorem}{Theorem}

\newtheorem{definition}[theorem]{Definition}
\newtheorem{example}[theorem]{Example}

\newtheorem{lemma}[theorem]{Lemma}

\newtheorem{proposition}[theorem]{Proposition}
\newtheorem{remark}[theorem]{Remark}

\newenvironment{proof}[1][Proof]{\textbf{#1.} }{\ \rule{0.5em}{0.5em}}
\input{tcilatex}
\begin{document}

\title{From Risk Measures to Research Measures}
\author{Marco Frittelli\thanks{%
Universit\`{a} degli Studi di Milano.}\qquad Ilaria Peri\thanks{%
Universit\`{a} degli Studi di Milano Bicocca.}}
\maketitle

\begin{abstract}
In order to evaluate the quality of the scientific research, we introduce a
new family of scientific performance measures, called Scientific Research
Measures (SRM). Our proposal originates from the more recent developments in
the theory of risk measures and is an attempt to resolve the many problems
of the existing bibliometric indices.

The SRM that we introduce are based on the whole scientist's citation record
and are: \textit{coherent}, as they share the same structural properties; 
\textit{flexible} to fit peculiarities of different areas and seniorities; 
\textit{granular, }as they allow a more precise comparison between
scientists, and \textit{inclusive}, as they comprehend several popular
indices.

Another key feature of our SRM is that they are planned to be \textit{%
calibrated} to the particular scientific community. We also propose a 
\textit{dual formulation} of this problem and explain its relevance in this
context.
\end{abstract}


\section{Introduction}

In the recent years the evaluation of the scientist's performance has become
increasingly important. In fact, most crucial decisions regarding faculty
recruitment, research projects, research time, academic promotion, travel
money, award of grants depend on great extent upon the scientific merits of
the involved researchers.

\bigskip

The scope of the valuation of the scientific research is twofold:

\begin{itemize}
\item Provide an updated picture of the existing research activity, in order
to allocate financial resources in relation to the scientific quality and
scientific production;

\item Determine an increase in the quality of the scientific research (of
the structures).
\end{itemize}

Even though both aims seems quite obvious, it is worthwhile to emphasize
that the selection of erroneous valuation criteria (one trivial example
would be "the number of the publications") could have important negative
impact on the future research quality. The methodologies for the valuation
can be divided into two categories:

\begin{itemize}
\item content valuation, based on internal judgments committee and external
reviews of peer panels.

\item context valuation, based on bibliometrics (i.e. statistics derived
from citation data) and the characteristics of the Journals associated to
the publications.
\end{itemize}

Economic considerations strongly depone of using the context method on a
systematic (yearly) base, while peer review is more plausible on a multiple
year base and should also be finalized to check, harmonize, and tune the
outcomes based on bibliometric indices.

Thanks also to the major availability of the online database (i.e. Google
Scholar, ISI Web of Science, MathSciNet, Scopus) several different
bibliometric measures have been recently introduced and applied.

There are several critics, as those clearly underlined by the Citation
Statistics Report of the International Mathematical Union (2008) \cite{CIT},
to the use of the citations as a key factor in the assessment of the quality
of the research. However, many of these critics can be satisfactorily
addressed and our proposal is one reasonable way to achieve this task.

\textbf{We agree that the quality} \textbf{of the scientific research can
not be reduced to citations, but we also believe that the information
embedded in citations should be properly quantified} \textbf{and should be
one component of the valuation of the quality of the scientific research}.

We emphasize that the output of the valuation is the classification of
authors (and structures) into few merit classes of homogeneous research
quality: it is not intended to provide a fine ranking. In the Appendix we
listed a brief summary of the pros and cons of bibliometric indices and of
the peer review process.

\bigskip

In 2005 Hirsch \cite{H05} proposed the \emph{h-index, }that is now the most
popular and used citation-based metric. The $h$-index of an author is
defined as the largest number $h\in \mathbb{N}$ satisfying the condition
that $h$ distinct publications of the author have (each one) $h$ citations.
The $h$-index is a vague attempt to measure at the same time the production
in terms of number of publications and the research quality in terms of
citations per publication. Our approach aims exactly to take better in
consideration the balance between these two components.

After its introduction, the $h$-index received wide attention from the
scientific community and it has been extended by many authors who proposed
other indices (for an overview see Alonso et al., 2009 \cite{ACHH}) in order
to overcome some of the drawbacks of it (see Bornmann and Daniel, 2007 \cite%
{BD07}).

\bigskip

In this paper we introduce three novel features in the methodology regarding
the measurement of the quality of scientific research:

\begin{enumerate}
\item The coherency of the research measures

\item The calibration technique

\item The dual setting.
\end{enumerate}

\bigskip

1. Differently from any existing approach, our formulation is clearly
germinated from the Theory of Risk Measures. The axiomatic approach
developed in the seminal paper by Artzner et al.\cite{ADEH99} turned out to
be, in this last decade, very influential for the theory of risk measures:
instead of focusing on some particular measurement of the risk carried by
financial positions (the variance, the $V@R$, etc. etc.), \cite{ADEH99}
proposed a class of measures satisfying some reasonable properties (the
\textquotedblleft coherent\textquotedblright\ axioms). Ideally, each
institution could select its own risk measure, provided it obeyed the
structural coherent properties. This approach added flexibility in the
selection of the risk measure and, at the same time, established a unified
framework.

\textit{We propose the same approach in order to determine a good class of
scientific performance measures, that we call Scientific Research Measures
(SRM).}

\bigskip

The theory of coherent risk measures was later extended to the class of
convex risk measures (Follmer and Schied \cite{FS02}, Frittelli and Rosazza 
\cite{FR02}). The origin of our proposal can be traced in the more recent
development of this theory, leading to the notion of quasi-convex risk
measures introduced by Cerreia-Vioglio et al. \cite{CMMMa} and further
developed in the dynamic framework by Frittelli and Maggis \cite{FM11}.
Additional papers in this area include: Cherny and Madan \cite{CM09}, that
introduced the concept of an Acceptability Index having the property of
quasi-concavity; Drapeau and Kupper \cite{DK10}, where the correspondence
between a quasi-convex risk measure and the associated family of acceptance
sets - already present in \cite{CM09} - is fully analyzed.

The representation of quasi-convex monotone maps in terms of family of
acceptance sets, as well as their dual formulations, are the key
mathematical tools underlying our definition of SRM.

\bigskip

2. \textit{A second feature of \ our approach is that our SRM are planned to
be "calibrated from the market data\textquotedblright }, a typical feature
of modeling in finance. As explained in Section 4, we \emph{calibrate} the
SRM from the historic data that are available for one particular scientific
area and seniority. In this way, each SRM will fit appropriately the
characteristics of the research field and seniority under consideration.

\bigskip

3. \textit{Our third innovation in this area, is provided by the dual
approach to the valuation of the quality of the scientific research}. As
explained in Section 3.2, we establish a duality between the primal space,
the space of random variables (representing the citations records) defined
on the set of Journals and its dual space, the space of the "Arrow-Debreu
price" of each Journal, which could be given by the impact factor of the
Journal. In section 3, we discuss this duality and show that our SRM fits
very well in this framework.

We finally report some empirical results obtained \ by\textit{\ }calibrating
the performance curves to a specific data set.

\bigskip

To summarize, we propose a family of SRMs that are:

\begin{itemize}
\item \textit{coherent}, as they share the same structural properties -
based on an axiomatic approach;

\item \textit{calibrated} to the particular scientific community;

\item \textit{flexible} in order to fit peculiarities of different areas and
ages;

\item \textit{robust}, as they can be defined, via duality, through a set of
probabilities representing the \textquotedblleft value\textquotedblright\ of
each Journal;

\item \textit{granular, }as they allow a more precise comparison between
scientists;

\item \textit{inclusive}, as they comprehends several popular indices.
\end{itemize}

\section{On a class of Scientific Research Measure}

We represent each author by a vector $X$ of citations, where the $i$-th
component of $X$\ represents the number of citations of the $i$-th
publication and the components of $X$ are ranked in decreasing order. We
consider the whole \emph{citation curve} of an author as a decreasing
bounded step functions $X$ (see Fig.\ref{author}) in the convex cone:%
\begin{eqnarray*}
\mathcal{X}^{+} &=&\left\{ 
\begin{array}{c}
X:\mathbb{R}\rightarrow \mathbb{R}_{+}\mid X\text{ is bounded, with only a
finite numbers of values,\ } \\ 
\text{decreasing on }\mathbb{R}_{+}\text{ and such that }X(x)=0\text{ for }%
x\leq 0.%
\end{array}%
\right\} 
\end{eqnarray*}%
\begin{figure}[h!]
\centering
\includegraphics[width=0.55\textwidth]{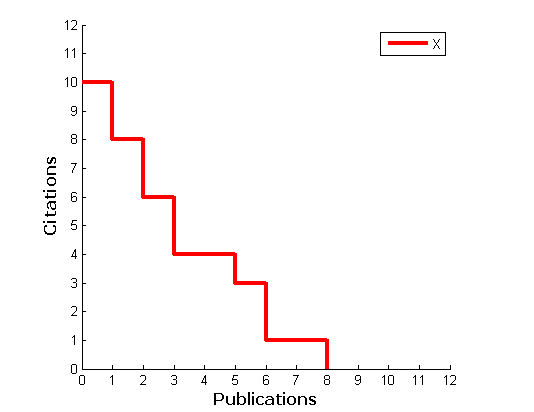}
\caption{Author's Citation Curve}\label{author}
\end{figure}


We compare the citation curve $X$ of an author with a theoretical citation
curve $f_{q}$ representing the desiderata citations at a fix performance
level $q$. For this purpose we introduce the following class of curves. Let $%
\mathcal{I}\subseteq \mathbb{R}$ be the index set of the \emph{performance
level}.\ For any $q\in \mathcal{I}$ we define the theoretical \emph{%
performance} \emph{curve of level }$\emph{q}$ as a function $f_{q}:\mathbb{R}%
\rightarrow \mathbb{R}_{+}$ that associates to each publication $x\in 
\mathbb{R}$ the corresponding number of citations $f_{q}(x)\in \mathbb{R}%
_{+} $.

\begin{definition}[Performance curves]
Given a index set $\mathcal{I}\subseteq \mathbb{R}$ of performance levels $%
q\in \mathcal{I}$, a class $\mathbb{F}:=\left\{ f_{q}\right\} _{q\in 
\mathcal{I}}$ of functions $f_{q}:\mathbb{R}\rightarrow \mathbb{R}_{+}$ is a 
\emph{family of performance curves} if

i) $\left\{ f_{q}\right\} $ is increasing in $q$, i.e. if $q\geq p$ then $%
f_{q}(x)\geq f_{p}(x)$ for all $x$;

ii) for each $q$, $f_{q}(x)$ is left continuous in $x$;

iii) $f_{q}(x)=0$ for all $x\leq 0$ and all $q.$
\end{definition}

The main feature of these curves is that a higher performance level implies
a higher number of citations. This family of curves is crucial for our
objective to build a SRM able to comprehend many of the popular indices and
to be calibrated to the scientific area and the seniority of the authors.

\begin{definition}[Performance sets and SRM]
Given a family of performance curves\emph{\ \ }$\mathbb{F=}\left\{
f_{q}\right\} _{q}$, we define the\emph{\ family of performance sets} $%
\mathcal{A}_{\mathbb{F}}:=\left\{ \mathcal{A}_{q}\right\} _{q}$ by%
\begin{equation*}
\mathcal{A}_{q}:=\left\{ X\in \mathcal{X}^{+}\mid X(x)\geq f_{q}(x)\text{
for all}\ x\in \mathbb{R}\right\} .
\end{equation*}%
The \emph{Scientific Research Measure (SRM) }is the map $\phi _{\mathbb{F}}:%
\mathcal{X}^{+}\rightarrow \mathbb{R}$ associated to $\mathbb{F}$ and $%
\mathcal{A}_{\mathbb{F}}$ given by%
\begin{eqnarray}
\phi _{\mathbb{F}}(X) &:&=\sup \left\{ q\in \mathcal{I}\mid X\in \mathcal{A}%
_{q}\right\}  \notag \\
&=&\sup \left\{ q\in \mathcal{I}\mid X(x)\geq f_{q}(x)\text{ for all}\ x\in 
\mathbb{R}\right\} .  \label{SRM}
\end{eqnarray}
\end{definition}

\begin{figure}[h!]
\centering
\includegraphics[width=0.5\textwidth]{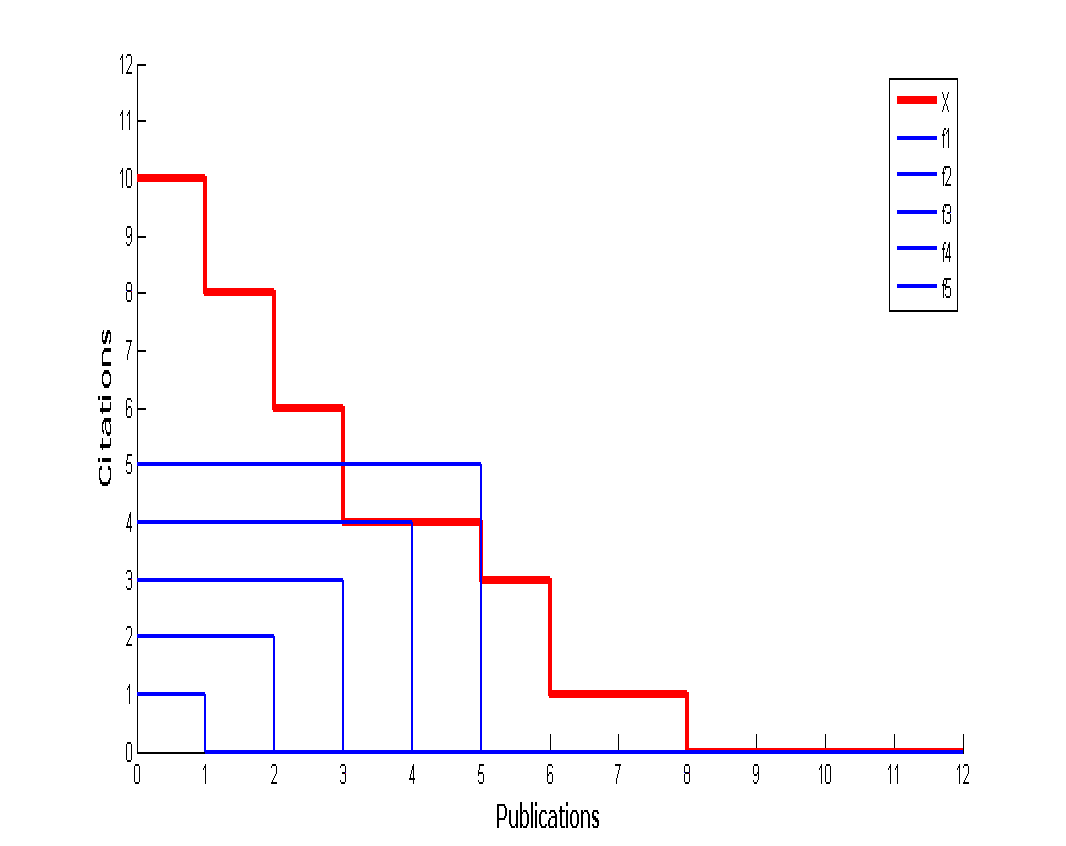}
\caption{Determination of a particular SRM. $h$-index equal to $4$.}\label{figSRM}
\end{figure}

The SRM $\phi _{\mathbb{F}}$ is obtained by the comparison between the real
citation curve of an author $X$ (the red line in Fig.\ref{figSRM}) and the
family $\mathbb{F}$ of performance curves (the blue line in Fig.\ref{figSRM}%
): the $\phi _{\mathbb{F}}(X)$ is the greatest level $q$ of the performance
curve $f_{q}$ below the author's citation curve $X.$%


\subsection{Some examples of existing SRMs}

The previous definition points out the importance of the family of
theoretical performance curves for the determination of the SRM. It is clear
that different choices of $\mathbb{F}:=\left\{ f_{q}\right\} _{q}$ lead to
different SRM $\phi _{\mathbb{F}}$. The following examples show that some
well known indices of scientific performance are particular cases of our
SRM. In the following examples, if $X$ has $p\geq 1$ publications that
received at least one citation, we set: $X=\sum_{i=1}^{p}x_{i}1_{(i-1,i]}$ ,
with $x_{i}\geq x_{i+1}\geq 1$ for all $i$, and $p$ satisfies $X=X1_{(0,p]}.$

\begin{example}[max \# of citations]
\label{cmax}The \emph{maximum number of citations} of the most cited
author's paper is the SRM $\phi _{\mathbb{F}_{c_{\max }}}$ defined by (\ref%
{SRM}), where the performance curves are: $f_{q}=q1_{(0,1]}$.%

\begin{figure}[h!]
\centering
\includegraphics[width=0.45\textwidth]{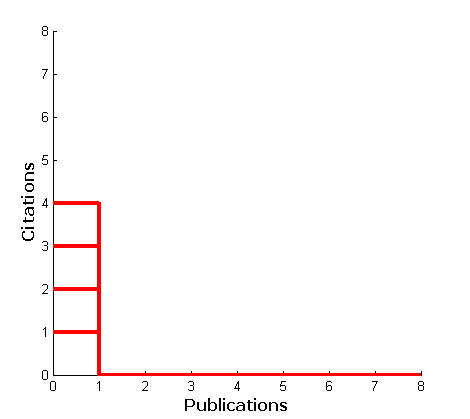}
\label{fcmax}
\end{figure}
\end{example}


\begin{example}[total number of publications]
\label{p}The \emph{total number of publications} with at least one citation
is the SRM $\phi _{\mathbb{F}_{p}}$ defined by (\ref{SRM}), where the
performance curves are: $f_{q}=1_{(0,q]}$.%
\end{example}

\begin{figure}[h!]
\centering
\includegraphics[width=0.45\textwidth]{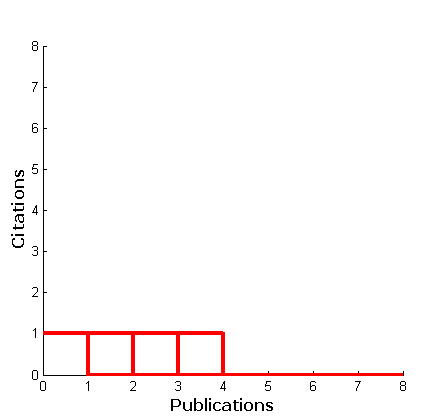}
\label{fp}
\end{figure}


\begin{example}[h-index]
\label{exh}The $h$-index defined by Hirsch \cite{H05} may be rewritten in
our setting. Indeed, the \emph{h-index} is the SRM $\phi _{\mathbb{F}_{h}}$
defined by (\ref{SRM}), where the performance curves are: $f_{q}=q1_{(0,q]}$.%
\end{example}

\begin{figure}[h!]
\centering
\includegraphics[width=0.45\textwidth]{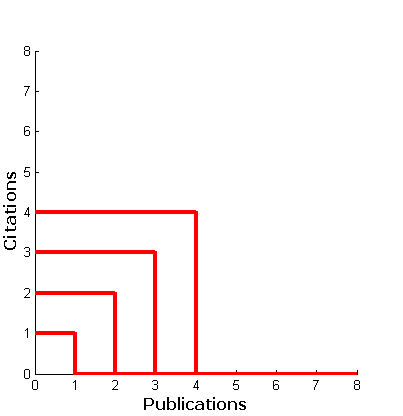}
\label{fh}
\end{figure}

%
\begin{example}[h$^{2}$-index]
\label{h^2Index}According to Kosmulski, 2006 \cite{K06} a scientist has 
\emph{h}$^{2}$-index \emph{q} if \emph{q} of his $n$ papers have at least $%
q^{2}$ citations each and the other $n-$\emph{q} papers have fewer than $%
q^{2}$ citations each. This index is the SRM $\phi _{\mathbb{F}_{h^{2}}}$%
defined by (\ref{SRM}), where the performance curves are: $%
f_{q}=q^{2}1_{(0,q]}$.
\end{example}

\begin{example}[h$_{\protect\alpha }$-index]
\label{haIndex}Eck and Waltman, 2008 \cite{EW06} proposed the \emph{h}$%
_{\alpha }$\emph{-index} as a generalization of the $h$-index defined as: "a
scientist has \emph{h}$_{\alpha }$-index \emph{h}$_{\alpha }$ if \emph{h}$%
_{\alpha }$ of his $n$ papers have at least $\alpha \cdot $\emph{h}$_{\alpha
}$ citations each and the other $n-$\emph{h}$_{\alpha }$ papers have fewer
than $\alpha \cdot $\emph{h}$_{\alpha }$ citations each". Hence, the \emph{h}%
$_{\alpha }$\emph{-index} is the SRM $\phi _{\mathbb{F}_{h_{\alpha }}}$
defined by (\ref{SRM}),where the performance curves are: $f_{q}=\alpha
q1_{(0,q]}$, $\alpha \in (0,\infty ).$
\end{example}

\begin{example}[w-index]
\label{windex}Woeginger, 2008 \cite{W0308} introduced the \emph{w-index}
defined as: "a \emph{w-index} of at least $k$ means that there are $k$
distinct publications that have at least $1$, $2$, $3$, $4$,..., $k$
citations, respectively". It is the SRM $\phi _{\mathbb{F}_{w}}$ defined by (%
\ref{SRM}), where the performance curves are: $f_{q}(x)=(-x+q+1)1_{(0,q]}.$
\end{example}

\begin{example}[h$_{rat}$-index \& h$_{r}$-index]
The rational and the real h-index, $h_{rat}$\emph{-index} and $h_{r}$\emph{%
-index, }introduced respectively\emph{\ }by Ruane and Tol, 2008 \cite{RT08}\
and Guns and Rousseau, 2009 \cite{GR09} are SRMs, indeed they could be
defined as the $h$-index but taking respectively $q\in \mathcal{I\subseteq }%
\mathbb{Q}$ and $q\in \mathcal{I\subseteq }\mathbb{R}$.
\end{example}

\subsection{Key properties of the SRM}

\begin{proposition}
\label{PROP}Let $\mathbb{F}$ be a family of performance curves, $\mathcal{A}%
_{\mathbb{F}}=\left\{ \mathcal{A}_{q}\right\} _{q}$ be the associated family
of performance sets and $\phi _{\mathbb{F}}$ be the associated \emph{SRM}$.$
Let\ $X_{1},X_{2}\in \mathcal{X}^{+}$. Then:

\begin{enumerate}
\item 
\begin{description}
\item[i)] $\left\{ \mathcal{A}_{q}\right\} $ is decreasing monotone: $%
\mathcal{A}_{q}\subseteq \mathcal{A}_{p}$ for any level $q\geq p$;

\item[ii)] for any $q,$ $\mathcal{A}_{q}$ is monotone: $X_{1}\in \mathcal{A}%
_{q}$ and $X_{2}\geq X_{1}$ implies $X_{2}\in \mathcal{A}_{q}$;

\item[iii)] for any $q$, $\mathcal{A}_{q}$ is convex: if $X_{1},X_{2}\in 
\mathcal{A}_{q}$ then $\lambda X_{1}+(1-\lambda )X_{2}\in \mathcal{A}_{q}$, $%
\lambda \in \left[ 0,1\right] $.
\end{description}

\begin{description}
\item[i)] $\phi _{\mathbb{F}}$ is monotone increasing: if $X_{1}\leq X_{2}$ $%
\Rightarrow \phi _{\mathbb{F}}(X_{1})\leq \phi _{\mathbb{F}}(X_{2})$;

\item[ii)] $\phi _{\mathbb{F}}$ is quasi-concave: $\phi _{\mathbb{F}%
}(\lambda X_{1}+(1-\lambda )X_{2})\geq \min (\phi _{\mathbb{F}}(X_{1}),$ $%
\phi _{\mathbb{F}}(X_{2}))$, $\lambda \in \lbrack 0,1]$.
\end{description}
\end{enumerate}
\end{proposition}

\begin{proof}

\begin{description}
\item[1)] The proof\ of the monotonicity and convexity of $\mathcal{A}_{%
\mathbb{F}}$ follows from the definition.

\item[2.i)] If $X_{1}\leq X_{2}$, then $X_{1}\geq f_{q}$ implies $X_{2}\geq
f_{q}.$ Hence $\left\{ q\in \mathcal{I}\mid X_{1}\geq f_{q}\right\}
\subseteq \left\{ q\in \mathcal{I}\mid X_{2}\geq f_{q}\right\} $.

\item[2.ii)] Let $\phi _{\mathbb{F}}(X_{1})\geq m$ and $\phi _{\mathbb{F}%
}(X_{2})\geq m$. By definition of $\phi _{\mathbb{F}}$, $\forall \varepsilon
>0$ $\exists q_{i}$ s.t. $X_{i}\geq f_{q_{i}}$\ and $q_{i}>\phi _{\mathbb{F}%
}(X_{i})-\varepsilon \geq m-\varepsilon $. Then $X_{i}\geq f_{q_{i}}\geq
f_{m-\varepsilon }$, as $\left\{ f_{q}\right\} _{q}$ is an increasing
family, and therefore $\lambda X_{1}+(1-\lambda )X_{2}\geq f_{m-\varepsilon
} $. As this holds for any $\varepsilon >0$, we conclude that $\phi _{%
\mathbb{F}}(\lambda X_{1}+(1-\lambda )X_{2})\geq m$ and $\phi _{\mathbb{F}}$
is quasi-concave.
\end{description}
\end{proof}

\bigskip

It is obviously reasonable that a SRM should be \emph{increasing:} if the
citations of a researcher $X_{2}$ dominate the citations of another
researcher $X_{1}$, publication by publication, then $X_{2}$ has a
performance greater than $X_{1}$.

\begin{example}
We show that a SRM is not in general quasi-convex, that is $\phi _{\mathbb{F}%
}(\lambda X_{1}+(1-\lambda )X_{2})\leq \max (\phi _{\mathbb{F}}(X_{1}),$ $%
\phi _{\mathbb{F}}(X_{2}))$ for all $\lambda \in \lbrack 0,1].$ Consider two
vectors, $X_{1}=[8\quad 6\quad 4\quad 2]$ and $X_{2}=[4\quad 2\quad 2\quad
2\quad 2]$, where $X_{2}$ has more publications than $X_{1}$ but less cited.
Computing, for example, the $w$-index we obtain that $\phi _{\mathbb{F}%
_{w}}(X_{1})=4$ and $\phi _{\mathbb{F}_{w}}(X_{2})=3$, while for the convex
combination $X=\frac{1}{2}X_{1}+\frac{1}{2}X_{2}=$ $\left[ 6\quad 4\quad
3\quad 2\quad 1\right] $ we have: $\phi _{\mathbb{F}_{w}}(X)=5$.
\end{example}

\subsection{Additional properties of SRMs}

We have seen that all the SRMs $\phi _{\mathbb{F}}$ share the same
structural properties of monotonicity and quasiconcavity. Of course other
relevant properties of $\phi _{\mathbb{F}}$ could be considered, which could
also be built in from the corresponding features of the family of
performance curves. In this section we show that this is the case for the
behavior of $\phi _{\mathbb{F}}$\ with respect to the addition of citations
(C-additivity) to existing papers.

\begin{definition}
A SRM $\phi _{\mathbb{F}}:\mathcal{X}^{+}\rightarrow \mathbb{R}$ is:

\begin{description}
\item[a)] $\mathrm{C}$-superadditive if $\phi _{\mathbb{F}}(X+m)\geq \phi _{%
\mathbb{F}}(X)+m$ for all $m\in \mathbb{R}_{+}$;

\item[b)] $\mathrm{C}$-subadditive if $\phi _{\mathbb{F}}(X+m)\leq \phi _{%
\mathbb{F}}(X)+m$ for all $m\in \mathbb{R}_{+}$;

\item[c)] $\mathrm{C}$-additive if $\phi _{\mathbb{F}}(X+m)=\phi _{\mathbb{F}%
}(X)+m$ for all $m\in \mathbb{R}_{+}$.
\end{description}
\end{definition}

\begin{definition}
A a family $\mathbb{F}=\left\{ f_{q}\right\} _{q}$ of performance curves is:

\begin{description}
\item[a)] \emph{slowly increasing in q} if $f_{q+m}-f_{q}\leq m$ for all $%
m\in \mathbb{R}_{+}$;

\item[b)] \emph{fast increasing in q} if $f_{q+m}-f_{q}\geq m$ for all $m\in 
\mathbb{R}_{+}$;

\item[c)] \emph{linear increasing in q} if $f_{q+m}-f_{q}=m$ for all $m\in 
\mathbb{R}_{+}$.
\end{description}
\end{definition}

These properties of the family of performance curves can be express in terms
of corresponding properties of the family $\mathcal{A}_{\mathbb{F}}$ of
performance sets.

\begin{lemma}
The family $\mathbb{F}$ of performance curve is slowly (resp. fast, linear)
increasing in $q,$ if and only if 
\begin{equation}
\mathcal{A}_{q}+m\subseteq \mathcal{A}_{q+m}\text{ (resp. }\mathcal{A}%
_{q+m}\subseteq \mathcal{A}_{q}+m\text{, }\mathcal{A}_{q+m}=\mathcal{A}%
_{q}+m)  \label{inclus1}
\end{equation}%
for all $m\in \mathbb{R}_{+}$ and $q\in \mathcal{I}.$
\end{lemma}

\begin{proof}
In order to show that $\mathcal{A}_{q}+m\subseteq \mathcal{A}_{q+m}$ we
observe that: 
\begin{align*}
\mathcal{A}_{q+m}& :=\left\{ X\mid X\geq f_{q+m}\right\} \\
\mathcal{A}_{q}+m& =\left\{ X+m\mid X\geq f_{q}\right\} \\
& =\left\{ X\mid X\geq f_{q}+m\right\} .
\end{align*}%
From $f_{q}+m\geq f_{q+m}$, we deduce that $X\geq f_{q}+m$ implies $X\geq
f_{q+m}$. Hence $X\in \mathcal{A}_{q}+m$ $\Longrightarrow $ $X\in \mathcal{A}%
_{q+m}.$

Regarding the other implication, we know that if $X\in \mathcal{A}_{q}+m$
then $X\in \mathcal{A}_{q+m}$, that is $X\geq f_{q}+m$ implies $X\geq
f_{q+m} $. This implies that $f_{q}+m\geq f_{q+m}$. Similarly for the other
cases.
\end{proof}

\bigskip

\begin{lemma}
If a family $\mathbb{F}$ of performance curves is slowly (resp. fast,
linear) increasing in $q,$ then $\phi _{\mathbb{F}}$ is $\mathrm{C}$%
-superadditive (resp. $\mathrm{C}$-subadditive, $\mathrm{C}$-additive).
\end{lemma}

\begin{proof}
In order to show that $\phi _{\mathbb{F}}(X+m)-m\geq \phi _{\mathbb{F}}(X)$
for all $m\in \mathbb{R}_{+}$ we use the definition in (\ref{SRM}) and we
observe that 
\begin{eqnarray*}
\phi _{\mathbb{F}}(X+m)-m &=&\sup \left\{ q\mid X+m\geq f_{q}\right\} -m \\
&=&\sup \left\{ q-m\mid X\geq f_{q}-m\right\} \\
&=&\sup \left\{ q\mid X\geq f_{q+m}-m\right\} .
\end{eqnarray*}%
Hence it's sufficient to show that $\left\{ q\mid \text{ }X\geq
f_{q}\right\} \subseteq \left\{ q\mid X\geq f_{q+m}-m\right\} $ and this is
true since $f_{q}\geq f_{q+m}-m$. Similarly for the other cases.
\end{proof}

\bigskip

As shown in the following examples, the reverse implication in the above
Lemma does not hold true.

\begin{example}
.

\begin{itemize}
\item The $h$-index in the example (\ref{exh}) is a $\mathrm{C}$-subadditive
SRM, but the associated family $\mathbb{F}$ of performance curves defined in
(\ref{fh}) is not fast increasing in $q$, nor slowly increasing. Indeed the
family is linear increasing on the Hirsch core $[0,h],$ but not outside it.

\item The same considerations hold for the $h^{2}$- and $h_{\alpha }$- index
(see examples (\ref{h^2Index}) and (\ref{haIndex})).

\item The family $\mathbb{F}$ defined in Example \ref{windex}, associated to
the $w$-index, is slowly increasing in $q$. This condition is sufficient to
say that the $w$-index is a $\mathrm{C}$-superadditive SRM.

\item The maximum number of citations of an article (see example \ref{cmax})
is a $\mathrm{C}$-additive SRM, even if the family $\mathbb{F}$ of
performance curves defined in \ref{fcmax} is not linear increasing in $q$.
This property holds only on $[0,1]$, since the performance curves are equal
to zero outside.

\item The total number of publications (see example \ref{p}) is a $\mathrm{C}
$-superadditive SRM since the family $\mathbb{F}$ of performance curves
defined in \ref{fp} is slowly increasing in $q$.
\end{itemize}
\end{example}

A further property concerns the addition of a single publication to the
author's citation record.

\begin{definition}
Let $p$ be the maximum number of publications with at least one citation$,$
so that $p$ satisfies: $X=X1_{(0,p]}.$ A SRM $\phi _{\mathbb{F}}:\mathcal{X}%
^{+}\rightarrow \mathbb{R}$\ is

\begin{description}
\item[a)] $\mathrm{P}$-superadditive if $\phi_{\mathbb{F}}(X+1_{\left\{
p+1\right\} })\geq\phi_{\mathbb{F}}(X)+1$;

\item[b)] $\mathrm{P}$-subadditive if $\phi_{\mathbb{F}}(X+1_{\left\{
p+1\right\} })\leq\phi_{\mathbb{F}}(X)+1$;

\item[c)] $\mathrm{P}$-additive if $\phi_{\mathbb{F}}(X+1_{\left\{
p+1\right\} })=\phi_{\mathbb{F}}(X)+1$;

\item[c)] $\mathrm{P}$-invariance if $\phi_{\mathbb{F}}(X+1_{\left\{
p+1\right\} })=\phi_{\mathbb{F}}(X)$.
\end{description}
\end{definition}

A SRM is $\mathrm{P}$-superadditive if the addition of a new publication
with one citation leads to an increase of the measure more than linear. Many
known SRMs are $\mathrm{P}$-invariance (i.e. the $c_{\max }$, $h$-, $h^{2}$-
and $h_{\alpha }$-index in the examples (\ref{cmax}) (\ref{exh}), (\ref%
{h^2Index}) and (\ref{haIndex})) as the addition of one citation to a new
publication leaves the SRM invariant. The $w$-index (in the example (\ref%
{windex})) is $\mathrm{P}$-subadditive as the addition of one citation to a
new publication makes it greater at most of $1$ unit. While the total number
of publications $p$ with at least one citation (in the example (\ref{p})) is
clearly $\mathrm{P}$-additive.

\section{On the Dual Representation of the SRM}

The goal of this section is to provide a dual representation of the SRM. To
this scope, we need some topological structure. Let $(\mathbb{R},\mathcal{B}(%
\mathbb{R}),\mu )$ be a probability space, where $\mathcal{B}(\mathbb{R})$
is the $\sigma $-algebra of the Borel sets, $\mu $ is a probability measure
on $\mathcal{B}(\mathbb{R})$. Since the citation curve of an author $X$ is a
bounded function, it appears natural to take $X\in L^{\infty }(\mathbb{R},%
\mathcal{B}(\mathbb{R}),\mu )$, where $L^{\infty }(\mathbb{R},\mathcal{B}(%
\mathbb{R}),\mu )$ is the space of $\mathcal{B}(\mathbb{R})$-measurable
functions that are $\mu $ almost surely bounded. If we endow $L^{\infty }$
with the weak topology $\sigma (L^{\infty },L^{1})$ then $L^{1}=(L^{\infty
},\sigma (L^{\infty },L^{1}))^{\prime }$ is its topological dual. In the
dual pairing $(L^{\infty },L^{1},\langle \cdot ,\cdot \rangle )$ the
bilinear form $\langle \cdot ,\cdot \rangle :L^{\infty }\times
L^{1}\rightarrow \mathbb{R}$ is given by $\langle X,Z\rangle =E[ZX]$, the
linear function $X\mapsto E[ZX]$, with $Z\in L^{1}$, is $\sigma (L^{\infty
},L^{1})$ continuous and $(L^{\infty },\sigma (L^{\infty },L^{1}))$ is a
locally convex topological vector space. In this framework, each element of
a performance family\emph{\ \ }$\mathbb{F=}\left\{ f_{q}\right\} _{q}$ is a $%
\mathcal{B}(\mathbb{R})$-measurable function, the inequalities between
random variables are meant to hold $\mu $-a.s., and we set:%
\begin{eqnarray}
\mathcal{A}_{q} &:&=\left\{ X\in L^{\infty }\mid X\geq f_{q}\right\} , 
\notag \\
\phi _{\mathbb{F}}(X) &:&=\sup \left\{ q\in \mathcal{I}\mid X\in \mathcal{A}%
_{q}\right\} .  \label{phi}
\end{eqnarray}%
We have seen in the Section 1 that the SRM is a quasi-concave and monotone
map. Under appropriate continuity assumptions, the dual representation of
these type of maps can be found in \cite{PV},\cite{Volle}, \cite{CMMMa}.

\begin{definition}
A map $\phi :L^{\infty }(\mathbb{R})\rightarrow \overline{\mathbb{R}}$ is $%
\sigma (L^{\infty },L^{1})-$upper-semicontinuous if the upper level sets%
\begin{equation*}
\left\{ X\in L^{\infty }(\mathbb{R})\mid \phi (X)\geq q\right\} 
\end{equation*}%
are $\sigma (L^{\infty },L^{1})-$closed for all $q\in \mathbb{R}$.
\end{definition}

\begin{lemma}
\label{Aqchiuso}If $\mathcal{A}_{\mathbb{F}}=\left\{ \mathcal{A}_{q}\right\}
_{q}$ is a family of performance sets then $\mathcal{A}_{q}$ is $\sigma
(L^{\infty },L^{1})$-closed for any $q$.
\end{lemma}

\begin{proof}
To prove that $\mathcal{A}_{q}$ is $\sigma (L^{\infty },L^{1})$-closed let $%
Y_{\alpha }\in \mathcal{A}_{q}$ be a net satisfying $Y_{\alpha }\overset{%
\sigma (L^{\infty },L^{1})}{\rightarrow }Y\in L^{\infty }$. By
contradiction, suppose that $\mu (B)>0$ where $B:=\left\{ Y<f_{q}\right\}
\in \mathcal{B}(\mathbb{R})$. Taking as a continuous linear functional $%
Z=1_{B}\in L^{1}$, from $Y_{\alpha }\overset{\sigma (L^{\infty },L^{1})}{%
\rightarrow }Y$ we deduce: $E[1_{B}f_{q}]\leq E[1_{B}Y_{\alpha }]\rightarrow
E[1_{B}Y]<E[1_{B}f_{q}]$.
\end{proof}

\bigskip

The following proposition shows the relation between the continuity property
of the family $\mathbb{F}$ of performance curves, those of the family $%
\mathcal{A}_{\mathbb{F}}$ of performance sets and those of the SRM $\phi _{%
\mathbb{F}}.$

\begin{proposition}
\label{contin}Let $\mathbb{F}$ be a family of performance curves. If $%
\mathbb{F}$ is \emph{left continuous in} $q$, that is 
\begin{equation*}
f_{q-\varepsilon }(x)\uparrow f_{q}(x)\text{ for }\varepsilon \downarrow 0,%
\text{ }\mu -a.s,
\end{equation*}%
then:

\begin{enumerate}
\item $\mathcal{A}_{\mathbb{F}}$ is left-continuous in $q$, that is%
\begin{equation*}
\mathcal{A}_{q}=\bigcap\limits_{\epsilon >0}\mathcal{A}_{q-\varepsilon }.
\end{equation*}

\item 
\begin{equation}
\mathcal{A}_{q}=\left\{ X\in L^{\infty }\mid \phi _{\mathbb{F}}(X)\geq
q\right\} \text{, for all }q\in \mathcal{I}.  \label{Aq=upplev}
\end{equation}

\item $\phi _{\mathbb{F}}$ is $\sigma (L^{\infty },L^{1})$%
-upper-semicontinuous.
\end{enumerate}
\end{proposition}

\begin{proof}
.

\begin{enumerate}
\item By assumption we have that $f_{q-\varepsilon }(x)\uparrow f_{q}(x)$
for $\varepsilon \rightarrow 0$, $\mu $-a.s.. We have proved in Proposition (%
\ref{PROP}) that $\left\{ \mathcal{A}_{q}\right\} $ is decreasing monotone,
hence we know that $\mathcal{A}_{q}\subseteq \bigcap\limits_{\varepsilon >0}%
\mathcal{A}_{q-\varepsilon }$. We need to prove that $\bigcap\limits_{%
\varepsilon >0}\mathcal{A}_{q-\varepsilon }\subseteq \mathcal{A}_{q}$. By
contradiction we suppose that there exist $X\in L^{\infty }$ such that $%
X\geq f_{q-\varepsilon }$ for every $\varepsilon >0$ but $X(x)<f_{q}(x)$ on
a set of positive measure $\mu $. Then there exist a $\delta >0$ such that $%
f_{q}(x)>X(x)+\delta $ on a measurable set $B$ with $b:=\mu (B)\in (0,1].$
Since $f_{q-\varepsilon }\uparrow f_{q}$ we may find $\overline{\varepsilon }%
>0$ such that $f_{q-\overline{\varepsilon }}(x)>f_{q}(x)-\frac{\delta }{2}$
on a measurable set $C$ with $\mu (C)>1-b$. Thus $\mu (B\cap C)>0$ and $%
X(x)\geq f_{q-\overline{\varepsilon }}(x)>f_{q}(x)-\frac{\delta }{2}>X(x)+%
\frac{\delta }{2}$ on $B\cap C$.

\item Now let 
\begin{equation*}
B_{q}:=\left\{ X\in L^{\infty }\mid \phi _{\mathbb{F}}(X)\geq q\right\} .
\end{equation*}%
$\mathcal{A}_{q}\subseteq B_{q}$ follows directly from the definition of $%
\phi _{\mathbb{F}}.$ We have to show that $B_{q}\subseteq \mathcal{A}_{q}.$
Let $X\in B_{q}$. Hence $\phi _{\mathbb{F}}(X)\geq q$ and, for all $%
\varepsilon >0,$ there exists $\overline{q}$ such that $\overline{q}>\phi _{%
\mathbb{F}}(X)-\varepsilon \geq q-\varepsilon $ and $X\geq f_{\overline{q}}.$
Since $\left\{ f_{q}\right\} _{q}$ is increasing in $q$ we have that $X\geq
f_{q-\varepsilon }$ for all $\varepsilon >0$, therefore $X\in \mathcal{A}%
_{q-\varepsilon }$. By item 1 and the left continuity in $q$ of the family $%
\mathbb{F}$ we know that $\left\{ \mathcal{A}_{q}\right\} $ is
left-continuous in $q$ and so: $X\in \bigcap\limits_{\epsilon >0}\mathcal{A}%
_{q-\varepsilon }=\mathcal{A}_{q}$.

\item By Lemma (\ref{Aqchiuso}) we know that $\mathcal{A}_{q}$ is $\sigma
(L^{\infty },L^{1})-$closed for any $q$ and therefore the upper level sets $%
B_{q}=\mathcal{A}_{q}$ are $\sigma (L^{\infty },L^{1})-$closed and $\phi _{%
\mathbb{F}}$ is $\sigma (L^{\infty },L^{1})$ upper semicontinuous.
\end{enumerate}
\end{proof}

The next lemma will be applied in the proof of theorem \ref{th}. It can be
proved in a way similar to the convex case (see for example \cite{FS05})

\begin{lemma}
\label{CFA}Let $\phi _{\mathbb{F}}:L^{\infty }\rightarrow \mathbb{R}$ be a
SRM. Then the following are equivalent:

\noindent $\phi _{\mathbb{F}}$ is $\sigma (L^{\infty },L^{1})$-upper
semicontinuous;

\noindent $\phi _{\mathbb{F}}$ is continuous from above: $X_{n},X\in
L^{\infty }$ and $X_{n}\downarrow X$ imply $\phi _{\mathbb{F}%
}(X_{n})\downarrow \phi _{\mathbb{F}}(X)$
\end{lemma}

\begin{proof}
Let $\phi _{\mathbb{F}}$ be $\sigma (L^{\infty },L^{1})$-upper
semicontinuous and suppose that $X_{n}\downarrow X$. The monotonicity of $%
\phi _{\mathbb{F}}$ implies $\phi _{\mathbb{F}}(X_{n})\geq \phi _{\mathbb{F}%
}(X)$ and $\phi _{\mathbb{F}}(X_{n})\downarrow $ and therefore $%
q:=\lim_{n}\phi _{\mathbb{F}}(X_{n})\geq \phi _{\mathbb{F}}(X)$. Hence $\phi
_{\mathbb{F}}(X_{n})\geq q$ and $X_{n}\in B_{q}:=\{Y\in L^{\infty }\mid \phi
_{\mathbb{F}}(Y)\geq q\}$ which is $\sigma (L^{\infty },L^{1})$-closed by
assumption. As the elements in $L^{1}$ are order continuous, from $%
X_{n}\downarrow X$ we get $X_{n}\overset{\sigma (L^{\infty },L^{1})}{%
\longrightarrow }X$ and therefore $X\in B_{q}.$ This implies that $\phi _{%
\mathbb{F}}(X)=q$ and that $\phi _{\mathbb{F}}$ is continuous from above.

Conversely, suppose that $\phi _{\mathbb{F}}$ is continuous from above. We
have to show that the convex set $B_{q}$ is $\sigma (L^{\infty },L^{1})$%
-closed for any $q$. By the Krein Smulian Theorem it is sufficient to prove
that $C:=B_{q}\cap \left\{ X\in L^{\infty }\mid \quad \parallel X\parallel
_{\infty }<r\right\} $ is $\sigma (L^{\infty },L^{1})$-closed for any fixed $%
r>0$. As $C\subseteq L^{\infty }\subseteq L^{1}$ and as the embedding 
\begin{equation*}
(L^{\infty },\sigma (L^{\infty },L^{1}))\hookrightarrow (L^{1},\sigma
(L^{1},L^{\infty }))
\end{equation*}%
is continuous it is sufficient to show that $C$ is $\sigma (L^{1},L^{\infty
})$-closed. Since the $\sigma (L^{1},L^{\infty })$ topology and the $L^{1}$
norm topology are compatible, and $C$ is convex, it is sufficient to prove
that $C$ is closed in $L^{1}$. Take $X_{n}\in C$ such that $X_{n}\rightarrow
X$ in $L^{1}$. Then there exists a subsequence $\left\{ Y_{n}\right\}
_{n}\subseteq \left\{ X_{n}\right\} _{n}$ such that $Y_{n}\rightarrow X$\
a.s. and $\phi _{\mathbb{F}}(Y_{n})\geq q$ for all $n$. Set $%
Z_{m}:=\sup_{n\geq m}Y_{n}\vee X$. Then $Z_{m}\in L^{\infty }$, since $%
\left\{ Y_{n}\right\} _{n}$ is uniformly bounded$,$ and $Z_{m}\geq Y_{m}$, $%
\phi _{\mathbb{F}}(Z_{m})\geq \phi _{\mathbb{F}}(Y_{m})$ and $%
Z_{m}\downarrow X$. From the continuity from above we conclude: $\phi _{%
\mathbb{F}}(X)=\lim_{m}\phi _{\mathbb{F}}(Z_{m})\geq \lim \sup_{m}\phi _{%
\mathbb{F}}(Y_{m})\geq q$. Thus $X\in B_{q}$ and consequently $X\in C.$
\end{proof}

\bigskip

When the family of performance curves $\mathbb{F}$ is left continuous,
Proposition (\ref{contin}) shows that the SRM is $\sigma (L^{\infty },L^{1})$%
-upper semicontinuous. Hence we can provide a dual representation for the
SRM in the same spirit of \cite{Volle}, \cite{CMMMa} and \cite{DK10}. In the
following theorem we first provide the representation of $\phi _{\mathbb{F}}$
in terms of the dual function $H$ defined in (\ref{HH}) and then we show
that $\phi _{\mathbb{F}}$ can also be represented in terms of the right
continuous version of $H$, which can be written in a different way as in (%
\ref{H+}). This dual representation will provide an interesting
interpretation of the SRM (see section 3.2).

Denote 
\begin{equation*}
\mathcal{P}:=\left\{ Q\ll P\right\} \text{ and }\mathcal{Z}:=\left\{ Z=\frac{%
dQ}{dP}\mid Q\in \mathcal{P}\right\} =\left\{ Z\in L_{+}^{1}\mid
E[Z]=1\right\} .
\end{equation*}

\begin{theorem}
\label{th}Suppose that the family of performance curves $\mathbb{F}$ is left
continuous. Each SRM $\phi _{\mathbb{F}}:L^{\infty }(\mathbb{R},\mathcal{B}(%
\mathbb{R}),\mu )\rightarrow \overline{\mathbb{R}}$\ defined in (\ref{phi})
can be represented as%
\begin{eqnarray}
\phi _{\mathbb{F}}(X) &=&\inf_{Z\in \mathcal{Z}}H(Z,E[ZX])=\inf_{Z\in 
\mathcal{Z}}H^{+}(Z,E[ZX])  \label{ff} \\
&=&\inf_{Q\in \mathcal{P}}H^{+}(Q,E_{Q}[X])\quad \text{\ for all }X\in
L^{\infty }  \notag
\end{eqnarray}%
where $H:L^{1}\times \mathbb{R}\rightarrow \overline{\mathbb{R}}$ is defined
by 
\begin{equation}
H(Z,t):=\sup_{\xi \in L^{\infty }}\left\{ \phi _{\mathbb{F}}(\xi )\mid
E[Z\xi ]\leq t\right\} ,  \label{HH}
\end{equation}%
$H^{+}(Z,\cdot )$ is its right continuous version: 
\begin{eqnarray}
H^{+}(Z,t) &:&=\inf_{s>t}H(Z,s)  \notag \\
&=&\sup \left\{ q\in \mathbb{R}\mid t\geq \gamma (Z,q)\right\} ,  \label{H+}
\end{eqnarray}%
and $\gamma :L^{1}\times \mathbb{R}\rightarrow \overline{\mathbb{R}}$ is
defined by: 
\begin{equation}
\gamma (Z,q):=\inf_{X\in L^{\infty }}\left\{ E[ZX]\mid \phi _{\mathbb{F}%
}(X)\geq q\right\} .  \label{Q}
\end{equation}
\end{theorem}

\begin{proof}
Step 1: $\phi _{\mathbb{F}}(X)=$ $\inf_{Z\in \mathcal{Z}}H(Z,E[ZX]).$

Fix $X\in L^{\infty }$. As $X\in \left\{ \xi \in L^{\infty }\mid E[Z\xi
]\leq E[ZX]\right\} $, by the definition of $H(Z,E[ZX])$ we deduce that, for
all $Z\in L^{1},$ 
\begin{equation*}
H(Z,E[ZX])\geq \phi _{\mathbb{F}}(X)
\end{equation*}%
hence 
\begin{equation}
\inf_{Z\in L^{1}}H(Z,E[ZX])\geq \phi _{\mathbb{F}}(X).  \label{infH}
\end{equation}%
We prove the opposite inequality. Let $\varepsilon >0$ and define the set%
\begin{equation*}
C_{\varepsilon }:=\left\{ \xi \in L^{\infty }\mid \phi _{\mathbb{F}}(\xi
)\geq \phi _{\mathbb{F}}(X)+\varepsilon \right\} 
\end{equation*}%
As $\phi _{\mathbb{F}}$ is quasi-concave and $\sigma (L^{\infty },L^{1})$%
-upper semicontinuous (Propositions \ref{PROP} and \ref{contin}), $C$ is
convex and $\sigma (L^{\infty },L^{1})-$closed. Since $X\notin
C_{\varepsilon }$, (if $\phi _{\mathbb{F}}(X)=-\infty ,$ we may take $%
C_{M}:=\left\{ \xi \in L^{\infty }\mid \phi _{\mathbb{F}}(\xi )\geq
-M\right\} $ and the following argument would hold as well) the Hahn Banach
theorem implies the existence of a continuous linear functional that
strongly separates $X$ and $C_{\varepsilon },$ that is there exist $%
Z_{\varepsilon }\in L^{1}$ such that%
\begin{equation}
E[Z_{\varepsilon }\xi ]>E[Z_{\varepsilon }X]\text{ for all }\xi \in
C_{\varepsilon }.  \label{CC}
\end{equation}%
Hence%
\begin{equation*}
\left\{ \xi \in L^{\infty }\mid E[Z_{\varepsilon }\xi ]\leq E[Z_{\varepsilon
}X]\right\} \subseteq C_{\varepsilon }^{c}:=\left\{ \xi \in L^{\infty }\mid
\phi _{\mathbb{F}}(\xi )<\phi _{\mathbb{F}}(X)+\varepsilon \right\} 
\end{equation*}%
and from (\ref{infH})%
\begin{align*}
\phi _{\mathbb{F}}(X)& \leq \inf_{Z\in L^{1}}H(Z,E[ZX])\leq H(Z_{\varepsilon
},E[Z_{\varepsilon }X]) \\
& =\sup \left\{ \phi _{\mathbb{F}}(\xi )\mid \xi \in L^{\infty }\text{ and }%
E[Z_{\varepsilon }\xi ]\leq E[Z_{\varepsilon }X]\right\}  \\
& \leq \sup \left\{ \phi _{\mathbb{F}}(\xi )\mid \xi \in L^{\infty }\text{
and }\phi _{\mathbb{F}}(\xi )<\phi _{\mathbb{F}}(X)+\varepsilon \right\}
\leq \phi _{\mathbb{F}}(X)+\varepsilon .
\end{align*}%
Therefore, $\phi _{\mathbb{F}}(X)=\inf_{Z\in L^{1}}H(Z,E[ZX])$. To show that
the $inf$ can be taken over the positive cone $L_{+}^{1}$, it is sufficient
to prove that $Z_{\varepsilon }\subseteq L_{+}^{1}$. Let $Y\in L_{+}^{\infty
}$ and $\xi \in C_{\varepsilon }.$ Given that $\phi _{\mathbb{F}}$ is
monotone increasing, $\xi +nY\in C_{\varepsilon }$ for every $n\in \mathbb{N}
$ and, from (\ref{CC}), we have:%
\begin{equation*}
E[Z_{\varepsilon }(\xi +nY)]>E[Z_{\varepsilon }X]\Rightarrow
E[Z_{\varepsilon }Y]>\frac{E[Z_{\varepsilon }(X-\xi )]}{n}\rightarrow 0,%
\text{ as }n\rightarrow \infty .
\end{equation*}%
As this holds for any $Y\in L_{+}^{\infty }$ we deduce that $Z_{\varepsilon
}\subseteq L_{+}^{1}$. Therefore, $\phi _{\mathbb{F}}(X)=\inf_{Z\in
L_{+}^{1}}H(Z,E[ZX])$.

By definition of $H(Z,t)$,%
\begin{equation*}
H(Z,E[ZX])=H(\lambda Z,E[(\lambda Z)X])\quad \forall Z\in L_{+}^{1}\text{ , }%
Z\neq 0,\text{ }\lambda \in (0,\infty ).
\end{equation*}%
Hence we deduce 
\begin{equation*}
\phi _{\mathbb{F}}(X)=\inf_{Z\in L_{+}^{1}(\mathbb{R})}H(Z,E[ZX])=\inf_{Z\in 
\mathcal{Z}}H(Z,E[ZX])=\inf_{Q\in \mathcal{P}}H(Q,E_{Q}[X]).
\end{equation*}%
Step 2: $\phi _{\mathbb{F}}(X)=\inf_{Z\in \mathcal{Z}}H^{+}(Z,E[ZX]).$

Since $H(Z,\cdot )$ is increasing and $Z\in L_{+}^{1}$ we obtain 
\begin{equation*}
H^{+}(Z,E[ZX]):=\inf_{s>E[ZX]}H(Z,s)\leq \lim_{X_{m}\downarrow
X}H(Z,E[ZX_{m}]),
\end{equation*}%
\begin{align*}
\phi _{\mathbb{F}}(X)& =\inf_{Z\in L_{+}^{1}}H(Z,E[ZX])\leq \inf_{Z\in
L_{+}^{1}}H^{+}(Z,E[ZX])\leq \inf_{Z\in L_{+}^{1}}\lim_{X_{m}\downarrow
X}H(Z,E[ZX_{m}]) \\
& =\lim_{X_{m}\downarrow X}\inf_{Z\in
L_{+}^{1}}H(Z,E[ZX_{m}])=\lim_{X_{m}\downarrow X}\phi _{\mathbb{F}%
}(X_{m})=\phi _{\mathbb{F}}(X),
\end{align*}%
where in the last equality we applied Lemma \ref{CFA} that guarantees the
continuity from above of $\phi _{\mathbb{F}}$. 

Step 3: $H^{+}(Z,t):=\inf_{s>t}H(Z,s)=\sup \left\{ q\in \mathbb{R}\mid
\gamma (Z,q)\leq t\right\} $ where $\gamma $ is defined in (\ref{Q}).

Denote 
\begin{equation*}
S(Z,t):=\sup \left\{ q\in \mathbb{R}\mid \gamma (Z,q)\leq t\right\} ,\text{ }%
(Z,t)\in L^{1}\times \mathbb{R},
\end{equation*}%
and note that $S(Z,\cdot )$ is the right inverse of the increasing function $%
\gamma (Z,\cdot )$ and therefore $S(Z,\cdot )$ is right continuous.\newline
To prove that $H^{+}(Z,t)\leq S(Z,t)$ it is sufficient to show that for all $%
p>t$ we have:

\begin{equation}
H(Z,p)\leq S(Z,p),  \label{50}
\end{equation}%
Indeed, if (\ref{50}) is true%
\begin{equation*}
H^{+}(Z,t)=\inf_{p>t}H(Z,p)\leq \inf_{p>t}S(Z,p)=S(Z,t),
\end{equation*}%
as both $H^{+}$ and $S$ are right continuous in the second argument.\newline
Writing explicitly the inequality (\ref{50}) 
\begin{equation*}
\sup_{\xi \in L^{\infty }}\left\{ \phi _{\mathbb{F}}(\xi )\mid E[Z\xi ]\leq
p\right\} \leq \sup \left\{ q\in \mathbb{R}\mid \gamma (Z,q)\leq p\right\}
\end{equation*}%
and letting $\xi \in L^{\infty }$ satisfying $E[Z\xi ]\leq p$, we see that
it is sufficient to show the existence of $q\in \mathbb{R}$ such that $%
\gamma (Z,q)\leq p$ and $q\geq \phi _{\mathbb{F}}(\xi )$. If $\phi _{\mathbb{%
F}}(\xi )=\infty $ then $\gamma (Z,q)\leq p$ for any $q$ and therefore $%
S(Z,p)=H(Z,p)=\infty $.

Suppose now that $\infty >\phi _{\mathbb{F}}(\xi )>-\infty $ and define $%
q:=\phi _{\mathbb{F}}(\xi ).$ As $E[\xi Z]\leq p$ we have:%
\begin{equation*}
\gamma (Z,q):=\inf \left\{ E[Z\xi ]\mid \phi _{\mathbb{F}}(\xi )\geq
q\right\} \leq p.
\end{equation*}%
Then $q\in \mathbb{R}$ satisfies the required conditions.

To obtain $H^{+}(Z,t):=\inf_{p>t}H(Z,p)\geq S(Z,t)$ it is sufficient to
prove that, for all $p>t,$ $H(Z,p)\geq S(Z,t)$, that is :%
\begin{equation}
\sup_{\xi \in L^{\infty }}\left\{ \phi _{\mathbb{F}}(\xi )\mid E[Z\xi ]\leq
p\right\} \geq \sup \left\{ q\in \mathbb{R}\mid \gamma (Z,q)\leq t\right\} .
\label{52}
\end{equation}

Fix any $p>t$\ and consider any $q\in \mathbb{R}$ such that $\gamma
(Z,q)\leq t$. By the definition of $\gamma $, for all $\varepsilon >0$ there
exists $\xi _{\varepsilon }\in L^{\infty }$ such that $\phi _{\mathbb{F}%
}(\xi _{\varepsilon })\geq q$ and $E[Z\xi _{\varepsilon }]\leq t+\varepsilon
.$ Take $\varepsilon $ such that $0<\varepsilon <p-t$. Then $E[Z\xi
_{\varepsilon }]\leq p$ and $\phi _{\mathbb{F}}(\xi _{\varepsilon })\geq q$
and (\ref{52}) follows.
\end{proof}

\begin{remark}[Interpretation of formulas \protect\ref{H+} and \protect\ref%
{Q}]
\label{rem}Let $Q$ be the 'weight' that we can assign to the author's
publications (for example, the impact factor of the Journal where the
article is published). For a fixed $Q,$ the term $\gamma (Q,q):=\inf \left\{
E_{Q}[\xi ]\mid \phi _{\mathbb{F}}(\xi )\geq q\right\} $ represents the
smallest $Q$-average of citations that a generic author needs in order to
have the SRM at least of $q$. We observe that this term is independent from
the citations of the author $X.$

On the light of these considerations we can interpret the term $%
H^{+}(Q,E_{Q}[X]):=\sup \left\{ q\in \mathbb{R}\mid E_{Q}[X]\geq \gamma
(Q,q)\right\} $ as the greatest performance level that the author $X$ can
reach, in the case that we attribute the weight $Q$ to the publications.
Namely, we compare the $Q$-average of the author $X$, $E_{Q}[X]$, with the
minimum $Q$-average necessary to reach each level $q$, that is $\gamma (Q,q)$%
.\ 
\end{remark}

\subsection{Examples}

In the following examples we find \textit{the dual representation of some
existing indices}. In all these examples the family $\mathbb{F}$ of
performance curves is left continuous hence, by Proposition (\ref{contin}),
the associated SRM $\phi _{\mathbb{F}}$ is $\sigma (L^{\infty },L^{1})$%
-upper semicontinuous and, from (\ref{Aq=upplev}), $X$ satisfies: 
\begin{equation*}
\phi _{\mathbb{F}}(X)\geq q\text{ iff }X\in \mathcal{A}_{q}\text{ iff }X\geq
f_{q}.
\end{equation*}%
Therefore, we find the dual representation computing $\gamma ,$ $H^{+}$ and $%
\phi _{\mathbb{F}}$ applying the formulas: (\ref{Q}),(\ref{H+}) and (\ref{ff}%
). Let $X\in L_{+}^{\infty }$, $Z\in L_{+}^{1}$, $q\in \mathbb{R}_{+}.$ Then:%
\begin{equation*}
\gamma (Z,q):=\inf_{\phi _{\mathbb{F}}(X)\geq q}E[ZX]=\inf_{X\geq
f_{q}}E[ZX]=E[Zf_{q}].
\end{equation*}%
Recall that $X=\sum_{i=1}^{p}x_{i}1_{(i-1,i]}$ , with $x_{i}\geq x_{i+1}>0$
for all $i,$ and that $p$ satisfies $X=X1_{(0,p]}\in L_{+}^{\infty }$.

\begin{example}[max \# of citations]
Consider the example (\ref{cmax}), where $f_{q}=q1_{(0,1]}$. Then:%
\begin{equation*}
\gamma (Z,q)=qE[Z1_{(0,1]}]
\end{equation*}%
and we obtain%
\begin{equation*}
H^{+}(Z,E[ZX]):=\sup \left\{ q\in \mathbb{R}\mid E[ZX]\geq
qE[Z1_{(0,1]}]\right\} =\frac{E[ZX]}{E[Z1_{(0,1]}]}
\end{equation*}%
(we may assume that $E[Z1_{(0,1]}]\neq 0$, otherwise $H^{+}(Z,E[ZX])=+\infty 
$ and it does not contribute to $\phi _{\mathbb{F}_{_{c_{\max }}}}$)$.$ In
our application, any non zero citation vector $X$ always satisfies $X\geq
x_{1}1_{(0,1]}$ and, since $E[X1_{(0,1]}]=x_{1}E[1_{(0,1]}],$ we also have: $%
\frac{X}{E[X1_{(0,1]}]}\geq \frac{1_{(0,1]}}{E[1_{(0,1]}]}$. Therefore, 
\begin{equation*}
E\left[ Z\frac{1_{(0,1]}}{E[1_{(0,1]}]}\right] \leq E\left[ Z\frac{X}{%
E[X1_{(0,1]}]}\right] \text{\quad }\forall Z\in L_{+}^{1}(%
\mathbb{R}
)
\end{equation*}%
and 
\begin{equation*}
\frac{E\left[ ZX\right] }{E\left[ Z1_{(0,1]}\right] }\geq \frac{E[1_{(0,1]}X]%
}{E[1_{(0,1]}]}\text{\quad }\forall Z\in L_{+}^{1}(%
\mathbb{R}
).
\end{equation*}%
Hence: 
\begin{align*}
\phi _{\mathbb{F}_{_{c_{\max }}}}(X)& =\inf_{Z\in L_{+}^{1}(%
\mathbb{R}
)}H^{+}(Z,E[ZX])=\inf_{Z\in L_{+}^{1}(%
\mathbb{R}
)}\frac{E[ZX]}{E[Z1_{(0,1]}]} \\
& =\frac{E[1_{(0,1]}X]}{E[1_{(0,1]}1_{(0,1]}]}=x_{1},
\end{align*}%
i.e. the infimum is attained at $Z=1_{(0,1]}\in L_{+}^{1}$, which is of
course natural as this SRM weights only the first publication.
\end{example}

\begin{example}[total \# of publications]
Consider the example (\ref{p}), where $f_{q}=1_{(0,q]}$. Then%
\begin{equation*}
\gamma (Z,q)=E[Z1_{(0,q]}]
\end{equation*}%
and we obtain 
\begin{equation*}
H^{+}(Z,E[ZX]):=\sup \left\{ q\in \mathbb{R}\mid E[ZX]\geq
E[Z1_{(0,q]}]\right\} .
\end{equation*}%
Hence the dual representation of the \emph{total number of publications}
with at least one citation is%
\begin{equation*}
\phi _{\mathbb{F}_{p}}(X)=\inf_{Z\in L_{+}^{1}(%
\mathbb{R}
)}\left\{ \sup_{E[ZX]\geq E[Z1_{[0,q]}]}q\right\} 
\end{equation*}%
We show indeed that $\phi _{\mathbb{F}_{p}}(X)=p$, where $p$ is
characterized by $X=X1_{(0,p]}$. For all $Z\in L_{+}^{1},$ and $q\leq p$ we
have 
\begin{equation*}
E[ZX]=E[ZX1_{(0,p]}]\geq E[1_{(0,q]}Z]
\end{equation*}%
and therefore 
\begin{equation*}
\sup_{E[ZX]\geq E[Z1_{(0,q]}]}q\geq p\qquad \forall Z\in L_{+}^{1},
\end{equation*}%
and $\phi _{\mathbb{F}_{p}}(X)\geq p$. Regarding the $\leq $ inequality, it
is enough to take $Z=1_{(p,p+\delta ]}$, with $\delta >0$. In this case, the
condition $E[ZX]\geq E[Z1_{(0,q]}]$ becomes%
\begin{equation*}
0=E[1_{(p,p+\delta ]}X]\geq E[1_{(p,p+\delta ]}1_{(0,q]}]
\end{equation*}%
that holds only for $q\leq p$, hence 
\begin{equation*}
H^{+}(Z,E[ZX])=\sup_{E[1_{(p,p+\delta ]}X]\geq E[1_{(p,p+\delta
]}1_{(0,q]}]}q=p
\end{equation*}%
and $\phi _{\mathbb{F}_{p}}(X)\leq p$.
\end{example}

\begin{example}[h-index]
Consider the example (\ref{exh}), where $f_{q}=q1_{(0,q]}.$ Then%
\begin{equation*}
\gamma (Z,q)=E[Zq1_{(0,q]}]
\end{equation*}%
and we obtain 
\begin{equation*}
H^{+}(Z,E[ZX]):=\sup \left\{ q\in \mathbb{R}\mid E[ZX]\geq
E[Zq1_{(0,q]}]\right\} .
\end{equation*}%
Hence the dual representation of the \emph{h-index} is%
\begin{equation*}
\phi _{\mathbb{F}_{h}}(X)=\inf_{Z\in L_{+}^{1}(%
\mathbb{R}
^{+})}\sup_{E[ZX]\geq E[Zq1_{(0,q]}]}q
\end{equation*}%
We indeed show that $\phi _{\mathbb{F}_{h}}(X)=h$, where $h$ is
characterized by $X1_{(0,h]}\geq h1_{(0,h]}$ and $X1_{(h,+\infty )}\leq
h1_{(h,+\infty )}$. First we check that $\phi _{\mathbb{F}_{h}}(X)\geq h$.
For all $Z\in L_{+}^{1},$ and $q\leq h$ we have 
\begin{equation*}
E[ZX]\geq E[ZX1_{(0,h]}]\geq E[Zh1_{(0,h]}]\geq E[Zq1_{(0,q]}],
\end{equation*}%
hence 
\begin{equation*}
\sup_{E[ZX]\geq E[Zq1_{(0,q]}]}q\geq h\qquad \forall Z\in L_{+}^{1}
\end{equation*}%
and $\phi _{\mathbb{F}_{h}}(X)\geq h$.

Regarding the $\leq $ side, take $Z=1_{(h,h+\delta ]}$ with $\delta >0$. For
all $q>h$ there exists $\delta >0$ such that $h+\delta <q$ and then 
\begin{equation*}
E[1_{(h,h+\delta ]}X]\leq E[1_{(h,h+\delta ]}h]<E[1_{(h,h+\delta
]}q1_{(0,q]}]
\end{equation*}%
hence 
\begin{equation*}
\sup_{E[1_{(h,h+\delta ]}X]\geq E[1_{(h,h+\delta ]}q1_{(0,q]}]}q\leq h
\end{equation*}%
and $\phi _{\mathbb{F}_{h}}(X)\leq h$.
\end{example}

\subsection{On the dual approach to SRM}

The dual representation in Theorem \ref{th} and the Remark \ref{rem} suggest
us another approach for the definition of a class of SRMs.

In other words, which is the interpretation of the duality that we are
discovering ?

The primal space is given by the set of all the possible author's citation
records, i.e. by all the random variables $X(w)$ defined on the events $w\in
\Omega $, where \textit{each event now corresponds to the journal in which
the paper appeared.}

The dual space is then represented by all possible linear valuation (the
"Arrow-Debreu price") of the journals.

We may fix a plausible family of probabilities $\mathcal{P}\subseteq \left\{
Q\ll P\right\} $ where each $Q(w)$ then represents the \emph{'value'}
attributed to the journal $w\in \Omega $. The valuation criterion for
journals (i.e. the selection of the family $\mathcal{P})$ has to be
determined a priori and could be based on the \emph{'impact factor'} or
other criterion. A specific $Q$ could attribute more importance to the
journals with a large number of citations (a large impact factor); another
particular $Q$ to the journals having a \textquotedblleft high
quality\textquotedblright . A priori there will be no consensus on the
selection of the family $\mathcal{P}$, hence a robust approach is needed.

\bigskip

As suggested from the dual representation results and in particular from the
equations (\ref{ff}) and (\ref{H+}) we consider, independently to the
particular scientist $X$, a family $\left\{ \gamma _{\beta }\right\} _{\beta
\in 
\mathbb{R}
}$ of functions $\gamma _{\beta }:\mathcal{P}\rightarrow \overline{\mathbb{R}%
}$ that associate to each $Q$ the value $\gamma _{\beta }(Q)$, that
represents the smallest $Q$-average of citations in order to reach a quality
index at least of $\beta $.

So given a particular value $Q(w_{i})$ for each $i^{th}$-journal and the
average citations $\gamma _{\beta }(Q)$ necessary to have an index level
greater than $\beta $, we build the SRM in the following way. We define the
function$\mathbb{\ }H^{+}:\mathcal{P}\times \mathbb{R}\rightarrow \overline{%
\mathbb{R}}$ that associates to each pair $(Q,E_{Q}(X))$ the number%
\begin{equation*}
H^{+}(Q,E_{Q}(X)):=\sup \left\{ \beta \in \mathbb{R}\mid E_{Q}(X)\geq \gamma
_{\beta }(Q)\right\} ,
\end{equation*}%
which represents the greatest quality index that the author $X$ can reach
when $Q$ is fixed, and we build the SRM as follows:%
\begin{equation*}
\phi (X):=\inf_{Q\in \mathcal{P}}H^{+}(Q,E_{Q}(X))
\end{equation*}%
which represents a prudential and robust approach with respect to $\mathcal{P%
}$, the plausible different selections of the evaluation of the Journals$.$
This SRM is by construction \emph{quasi-concave} and \emph{monotone
increasing}. Theorem \ref{th} exhibits the relationship between the
performance curve approach and this dual approach.

\section{Example of the calibration of a SRM}

Since the SRM introduced in Section 2 depends on the particular family $%
\mathbb{F}$ of performance curves, in this section we provide an example
showing how to \emph{calibrate} the family $\mathbb{F}$ from the historic
data available for one particular scientific area and seniority. In this
way, the SRM will fit appropriately the characteristics of the research
field and seniority under consideration. We recall that the SRM should be
used only in \emph{relative} terms (to compare the author quality with
respect to the other researchers in the same area) in order to classify the
authors (and structures) into few classes of homogeneous research quality.

\subsection{Determination of the family $\left\{ f_{q}\right\} _{q}$ and of
the SRM}

The first step consists in the selection of a representative sample of $M$
authors in the same scientific area and with the same seniority and then
from this sample of authors we need to extrapolate the family of curves $%
\left\{ f_{q}\right\} _{q}$ that better represents the citation curve of the
area and seniority. The analysis of the citation vectors of each author (see
Fig.\ref{curve}) shows that the theoretical model may be described (for this
particular scientific area) by the formula%
\begin{equation}
f_{q}(x)=\frac{q}{x^{\beta }}  \label{hyp}
\end{equation}%

\begin{figure}[h!]
\centering
\includegraphics[width=0.75\textwidth]{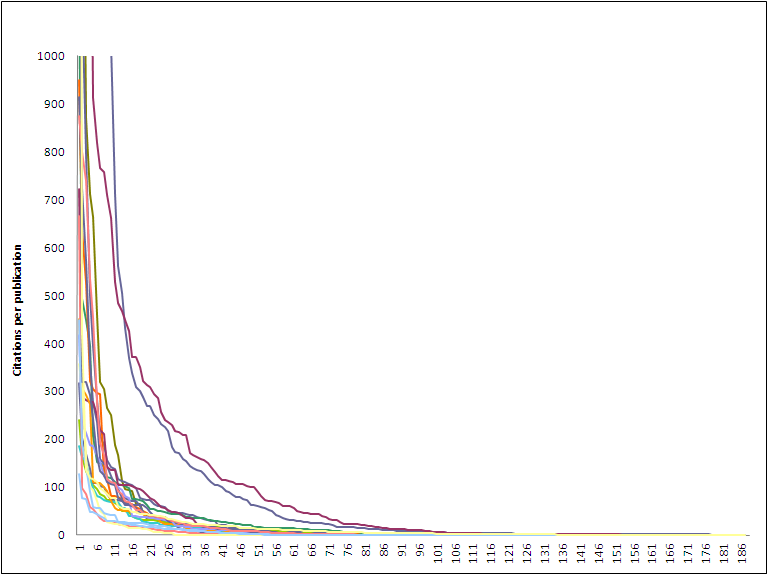}
\caption{Citation curves of 20 senior authors in Math Finance area.}\label{curve}
\end{figure}

with $q,\beta \in 
\mathbb{R}
_{+}$. Setting $\ln f_{q}=Y$, $\ln (q)=\widehat{q},$ $\ln x=X$, $\beta =%
\widehat{\beta }$ we obtain the linearized model%
\begin{equation}
Y=\widehat{q}-\widehat{\beta }X\text{.}  \label{ret}
\end{equation}

For each $i$-th author of the sample we determine $\widehat{\beta }_{i}$
that minimizes the sum of the square distances of the points from the line (%
\ref{ret}). Then, we compute $\bar{\beta}$ as the average of the $\hat{\beta}%
_{i}$:%
\begin{equation*}
\bar{\beta}=\frac{1}{M}\tsum\limits_{i=1}^{M}\hat{\beta}_{i}\text{.}
\end{equation*}%
Once the parameter $\overline{\beta }$ is fixed, we obtain the family of
performance curves $f_{q}(x)=\frac{q}{x^{\overline{\beta }}}$ and then the
associated SRM (hereafter called the $\phi $-index) is: 
\begin{equation}
\phi (X)=\sup \left\{ q\in \mathbb{R}\mid X(x)\geq \frac{q}{x^{\overline{%
\beta }}}\text{\qquad }\forall x\text{ }\right\}   \label{FIINDEX}
\end{equation}

\subsection{The empirical results}

We have chosen a group of\ $20$ well established researchers in the
mathematical finance area. We have computed the $\widehat{\beta }_{i}$ for
each author and we have found that $\overline{\beta }=1,62$.

In the following table (Fig.\ref{CompIndices}.a) we report the results and
the respective ranking obtained calculating the $\phi $-index as in (\ref%
{FIINDEX}) and the $h$-index for each author. Fig.\ref{CompIndices}.b shows
that the hyperbole-type curve (red line) corresponding to the author's $\phi 
$-index is always below his citation curve (blue line), in the domain $(0,p)$%
.%
\begin{figure}[h!]
\centering
\includegraphics[width=0.45\textwidth]{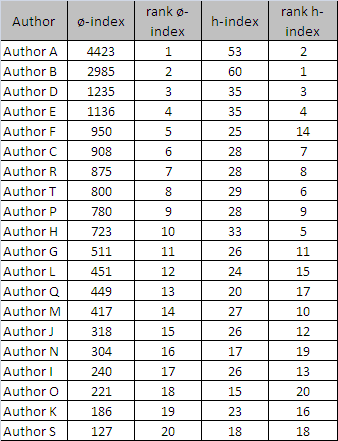}
\includegraphics[width=0.45\textwidth]{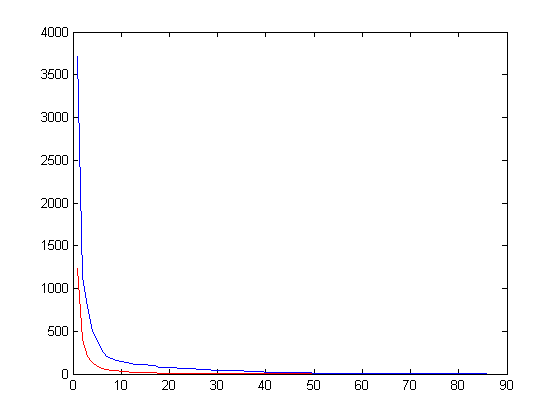}
\caption{(a) left - (b) right}\label{CompIndices}
\end{figure}


Notice that the author $F$ increases his index, from the 14th position in
the $h$-index to the 5th one of the $\phi $-index. If we compare this author
with the author $I$, we note that both have almost the same $h$-index but
the $\phi $-index of $F$ is much greater than the $\phi $-index of $I$.
Analyzing their citation curves we observe that they have the same number of
publications, but $F$ has in general a lot more citations for any
publication than $I$, especially those in the Hirsch-core. The same reasons
explain the different ranking of the authors $H$ and $D$.

The conclusion is that the calibrated family of performance curves $\mathbb{F%
}$ takes in the correct account the balance between the number of
publications and the citations per publication, a characteristic indeed of
the specific area under consideration.

\section{Appendix}

We provide a brief summary of the positive and negative features of the peer
review process and of the bibliometric indices. We are well aware that these
remarks are incomplete and represent a \textit{subjective, not
scientifically based, view} of this complex and controversial theme. In the
references, as outlined in the introduction, more detailed arguments can be
found. The following remarks however shows that many possible drawbacks of
bibliometrics indices may be smoothened and reduced by an appropriate use of
them and by the selection of a more convenient class of indices of the type
we presented in the previous sections.

\subsection{Summary of the pros and cons of context evaluation (bibliometric
indices)}

\noindent \textbf{Pros}:

\begin{itemize}
\item \textbf{Easily accessible,} from the online databases (Google Scholar,
ISI Web, MathSciNet, Scopus...);

\item \textbf{Not expensive}: can be used systematically, especially if
tested - every $n$ years - with peer review.

\item \textbf{Quick to compute}

\item \textquotedblleft \textbf{Objective}\textquotedblright , in the
reductive meaning of being independent from individual judgements.
\end{itemize}

\noindent \textbf{Cons}:

\begin{itemize}
\item \textbf{Subjective interpretation of citations}, as it can be more
subjective than the judgment of experts - see Citation Statistics Report of
the International Mathematical Union (2008) \cite{CIT}.

\begin{itemize}
\item The new metric must be validated against other (possibly non metric)
criterion already validated.

\item It has been pointed out - see the discussion in the American Scientist
Open Access Forum, 2008 \cite{ASOAF}- that citation metrics are extremely
correlated with peer reviews.
\end{itemize}

\item \textbf{Improper comparison} of papers belonging to different fields.

\begin{itemize}
\item The SRM should be used to rank each author inside his scientific
community (e.g.: top 10\% - top 30\% \ - average...). It provides \textit{%
relative }- to fields - values, not absolute values. However, this allows
also for a coarse comparison of authors belonging to different areas, in the
sense that it is possible to easily recognize the authors that are in the
same (top/lower/ ...) merit class in each area.
\end{itemize}

\item \textbf{Improper comparison of papers having different ages.}

\begin{itemize}
\item Our SRM may be calibrated to different ages (as well as different
areas).
\end{itemize}

\item \textbf{Different databases provide different citations.}

\begin{itemize}
\item Many areas (naturally) share the same database.

\item The outcome of the scientific measure is in relative terms: the
ranking of one author is compared with the ranking of all researchers in the
same area (hence using the same database).

\item Different databases (Google Scholar, MathSciNet,...) provide different
numbers (in terms of citation of each paper), but only via a scaling factor:
the overall ranking of the papers, with respect to the number of citations
received, remains essentially the same, see \cite{ASOAF}.
\end{itemize}

\item \textbf{Co-authors}

\begin{itemize}
\item It is possible to normalize the citation numbers per each single
author. For some fields (where papers have typically many co-authors) this
may be problematic.
\end{itemize}

\item \textbf{Incorrect citations} attributed to an author and self citations

\begin{itemize}
\item Both problems can be easily addressed by the systematic use of Author
Codes (a code that identify the author).
\end{itemize}

\item \textbf{A single number is insufficient for the evaluation of a
complex feature}, such as scientific research.

\begin{itemize}
\item We agree: It is necessary to find multiple metrics (including
time-based metrics). We propose one of them.

\item This argument should not lead to abandon the search of appropriate
multiple metrics.
\end{itemize}

\item \textbf{Quality} of the scientific research \textbf{can not be reduced
to citations}

\begin{itemize}
\item Agree: it is only one component that however should be properly
quantified.
\end{itemize}

\item \textbf{Negative credit:} citations may be attributed \textit{not as
reward citations} (to give credit to the work of the cited author) but as
negative credit (or \textquotedblleft rhetorical credit\textquotedblright\
due to the prestige of the cited author).

\begin{itemize}
\item True. Many are the motivations of citations and they varies among
authors: they do not always reflects reward, but certainly a large
percentage of citations are credit ones. Indeed:

\item The fact that citation based statistics often agree with other
validated form of valuation (peer review), see \cite{ASOAF}, suggests that,
to some degree, these metrics indeed reflects the impact of the author's
research.

\item The periodical peer review valuation should point out the macroscopic
exceptions to reward citations (papers mostly cited for their fallacy).
\end{itemize}

\item \textbf{Disincentive} for young researcher to study subjects more
innovative but less popular

\begin{itemize}
\item True, even though this could be compensated by the consideration that
innovative paper (in a new field) typically receive many citations.
\end{itemize}

\item \textbf{Negative Implications}: The use of citation based metrics will
increase the number of citations (and improper ones).

\begin{itemize}
\item The abuse of citations is comparable with intentional misjudgment by
referee: unfortunately this is always possible.

\item When citations number are high (in the order of hundreds) it is
difficult to modify the citation records with self or friendly citations.

\item It is not completely unfair that a strong scientific group (capable to
produce a large number of published papers) receives additional credit (due
to potential additional citations from the group).
\end{itemize}
\end{itemize}

\subsection{Summary of the pros and cons of content valuation (peer review)}

\noindent \textbf{Pros:}

\begin{itemize}
\item \textbf{effective assessment} of the quality of the research;
\end{itemize}

\noindent \textbf{Cons:}

\begin{itemize}
\item \textbf{expensive,} in term of time and people involved: It \textit{%
can not be used systematically.}

\item \textbf{subjective,} since the result depends on the referees: do they
operate properly, are they competent and reliable? The choice of the
referees is a very delicate issue.

\item \textbf{non-uniformity of the judgment}, as each evaluator has a
personal scaling preferences leading to different ranking (specially in
different areas).
\end{itemize}

\end{document}